\documentclass{amsart}

\newtheorem{theorem}{Theorem}[section]
\newtheorem{prop}{Proposition}[section]

\newtheorem{lemma}[theorem]{Lemma}

\theoremstyle{definition}
\newtheorem{definition}[theorem]{Definition}

\theoremstyle{remark}
\newtheorem{remark}[theorem]{Remark}

\numberwithin{equation}{section}

\newcommand{\abs}[1]{\lvert#1\rvert}
\DeclareMathOperator{\re}{Re}
\DeclareMathOperator{\im}{Im}
\newcommand{\ud}{\mathrm{d}}

\begin{document}

\title{Continuity of spectral averaging}

\author{C.A.Marx}
\address{Department of Mathematics, University of California, Irvine
CA, 92717}
\email{cmarx@math.uci.edu}
\thanks{The author was supported by NSF Grant DMS - 0601081.}


\begin{abstract}
We consider averages $\kappa$ of spectral measures of rank one perturbations with respect to a $\sigma$-finite measure $\nu$. It is examined how various degrees of continuity of $\nu$ with respect to $\alpha$-dimensional Hausdorff measures ($0 \leq \alpha \leq 1$) are inherited by $\kappa$. This extends Kotani's trick where $\nu$ is simply the Lebesgue measure. 
\end{abstract}

\maketitle

\section{Introduction}
Let $A$ be a bounded self adjoint operator on a separable Hilbert space $\mathcal{H}$. Fix a normalized vector $\phi \in \mathcal{H}$. Consider the family of rank one perturbations
\begin{equation} \label{eq_defrk1}
A_{\lambda} := A + \lambda \langle \phi , . \rangle \phi ~\mbox{,} 
\end{equation} 
indexed by the real parameter $\lambda$. Despite its simple form, the family in (\ref{eq_defrk1}) proves to be a very useful tool in the study of discrete random Schr\"odinger operators. There, rank one perturbations correspond to fluctuations of the potential at a lattice site. Ref. \cite{A} summarizes several of these applications, among them the Simon-Wolf criterion for spectral localization, the theory of Aizenman-Molchanov for the Anderson model, and Wegner's estimate. 

Crucial to most of these applications is a result known as spectral averaging or Kotani's trick. It allows to relate the spectral behavior for fixed values of $\lambda$ to the spectral properties inherent to the entire family $\{A_\lambda\}$, i.e. upon a  variation of $\lambda$.

Denote by $\mathrm{d}\mu(x)$ and $\mathrm{d}\mu_{\lambda}(x)$ the spectral measure with respect to $\phi$ for the operator $A$ and $A_{\lambda}$, respectively. Kotani's trick is the following result:

\begin{theorem}[Kotani's trick] \label{thm_kotani}
Let $B$ be a Borel set on the real line. Then, 
\begin{equation} \label{eq_specavag}
\vert B \vert = \int \mu_{\lambda}(B) \mathrm{d}\lambda ~\mbox{.}
\end{equation}
Here, $\vert . \vert$ denotes the Lebesgue measure.
\end{theorem}

Different proofs and applications of this result were given in \cite{B,C,C1,C2, C3, C4, D,E}. We note that for some purposes, among them the Simon-Wolff criterion, a weaker formulation is sufficient; this weaker result states that the Borel measure on the right hand side of (\ref{eq_specavag}) is absolutely continuous with respect to Lebesgue. In  fact, in the original proof of the Simon-Wolff criterion (see \cite{E}, Theorem 5) the authors show that the measure 
\begin{equation} \label{eq_simonwolff}
\kappa(.) = \int \mu_\lambda(.) \frac{1}{1+\lambda^2} \mathrm{d}\lambda ~\mbox{,}
\end{equation}
is mutually equivalent to the Lebesgue measure. 

Eq. (\ref{eq_simonwolff}) suggests the following generalization: For $\nu$ a $\sigma$-finite Borel measure on $\mathbb{R}$, define a measure $\kappa$ by
\begin{equation} \label{eq_defkappa}
\kappa(.) = \int \mu_\lambda(.) \mathrm{d}\nu(\lambda) ~\mbox{.}
\end{equation}
Such averages were first considered in \cite{F} for a finite measure $\nu$. There, relation (\ref{eq_relBorel}) was discovered for a finite measure $\nu$ and used to estimate the Hausdorff dimension of set $\{\lambda : A_\lambda ~\mbox{has some continuous spectrum}\}$ (see Theorem 5.2 therein). 

In view of the measure defined in (\ref{eq_defkappa}), Kotani's trick ($\mathrm{d}\nu(\lambda)= \mathrm{d}\lambda$) and the result for $\mathrm{d}\nu(\lambda)=\frac{1}{1+\lambda^2} \mathrm{d}\lambda$ in (\ref{eq_simonwolff}) become statements about continuity properties of the measure $\nu$ being inherited by $\kappa$. 

In this article we pursue this continuity based approach to spectral averaging. We will show how various degrees of continuity of $\nu$ with respect to $\alpha$-dimensional Hausdorff measures ($\alpha \leq 1$) are inherited by $\kappa$. For a definition of $\alpha$ continuity see definition \ref{def_alphacont}. Our main result is the following Theorem:

\begin{theorem} \label{thm_main}
If $\nu$ is absolutely continuous with respect to Lebesgue, so is $\kappa$. Additionally, if $\nu$ is $\alpha$c, $0 < \alpha < 1$, then $\kappa$ is $\delta$c for all $\delta<\alpha$.
\end{theorem}
Kotani's trick then arises as a special case, where the density of $\ud\kappa(x)$ can be calculated explicitly. 

The paper is organized as follows: Sec. \ref{sec_boreltraf} summarizes some results of the theory of Borel transforms and rank one perturbations as needed for the further development. After showing that mere continuity of $\nu$ is inherited by $\kappa$ (Theorem \ref{thm_cont}), we examine the situation for $\nu$ being uniformly-$\alpha$-H\"older continuous (see definition \ref{def_UHc}). In particular, we shall show that uniform 1-H\"older continuity of $\nu$ is inherited by $\kappa$. Kotani's trick follows as special case if $\ud \nu(x) = \ud x$. Finally, in Sec. \ref{sec_conthd} and \ref{sec_main} we employ the Rogers-Taylor decomposition of measures with respect to Hausdorff measures to prove theorem \ref{thm_main}.

I wish to thank my advisor Svetlana Jitomirskaya for encouragement and valuable discussions during the preparation of this manuscript.

\section{Borel transforms \& rank one perturbations} \label{sec_boreltraf}

The key quantity to understand the spectral properties of the family (\ref{eq_defrk1}) is the Borel transform of the spectral measure $\mathrm{d}\mu$ associated with the {\em{unperturbed}} operator $A$ and the vector $\phi$. In general, if $\eta$ is a Borel measure with $\int \frac{1}{1+\vert y \vert} \mathrm{d}\eta(y) < \infty$, for $z \in \mathbb{H}^{+}$ we define
\begin{equation} \label{eq_Boreltr}
F_\eta(z) := \int \dfrac{\mathrm{d}\eta(y)}{y-z} ~\mbox{,}
\end{equation}
the Borel transform of the measure $\eta$. Letting $z=x+ i \epsilon ~\mbox{,} ~\epsilon>0$, we may split $F_\eta(x+i\epsilon)$ into its real and imaginary part, i.e. $F_\eta(x+i \epsilon)=:Q_\eta(x+i \epsilon) + i P_\eta(x+i \epsilon)$, where
\begin{eqnarray*}
Q_\eta(x + i \epsilon) & = & \int \dfrac{y-x}{(y-x)^2 + \epsilon ^2}  \ud \eta(y) ~\mbox{,} \\
P_\eta(x + i \epsilon) & = & \int \dfrac{\epsilon}{(y-x)^2 + \epsilon ^2}  \ud \eta(y) ~\mbox{.}
\end{eqnarray*}
We shall refer to $P_\eta$ and $Q_\eta$ as Poisson- and  conjugate Poisson transform of the measure $\eta$, respectively. Whereas $Q_\eta(x+i \epsilon)$ depends on the ``symmetry'' of $\eta$ around $x$, $P_\eta(x + i \epsilon)$ carries information about the growth of the measure $\eta$ at $x$. A detailed analysis about the asymptotic behavior of $P_\eta$ and $Q_\eta$ is given in \cite{H}. 

The relation between the local growth of a measure and its Poisson transform follows from the following simple estimate: Given $\alpha \in [0,1]$, then for $x \in \mathbb{R}$ and $\epsilon > 0$
\begin{eqnarray} \label{eq_estimategrwoth}
& \epsilon^{1-\alpha} P_{\eta}(x+i \epsilon) \geq \epsilon^{1-\alpha} \int_{(x-\epsilon , x+\epsilon)} \dfrac{\epsilon}{(x-y)^2 + \epsilon ^2} \ud \eta(y)  \geq \frac{1}{2 \epsilon^\alpha} M_{\eta}(x;\epsilon) ~\mbox{,}
\end{eqnarray}
where $M_{\eta}(x;\epsilon) := \eta(x-\epsilon,x+\epsilon)$ denotes the growth function of $\eta$ at $x$. 

\begin{remark} 
As will be seen below (see Theorem \ref{thm_charact}), it is useful to consider the Poisson transform of a measure even if its Borel transform does not exist. A necessary and sufficient condition for $P_\eta(x + i \epsilon) < \infty$, $\epsilon > 0$, is $\int \frac{1}{1+x^2} \ud \eta(x) < \infty$. 
\end{remark}

For $\alpha \geq 0$ we define,
\begin{eqnarray}
\overline{D}_{\eta}^{\alpha}(x) & := & \limsup_{\epsilon \to 0^+} \dfrac{\eta(x-\epsilon , x+\epsilon)}{ \epsilon^\alpha} ~\mbox{,} 
\end{eqnarray}
the upper-$\alpha$-derivative of a measure $\eta$ at a point $x \in \mathbb{R}$.

Above estimate (\ref{eq_estimategrwoth}) leads to the following result proven e.g. in \cite{F}:
\begin{prop} \label{prop_deriv}
Let $\alpha \in [0,1]$ and $x \in \mathbb{R}$ be fixed. Then $\overline{D}_{\eta}^{\alpha}(x)$ and \newline $\limsup_{\epsilon \to 0^+} \epsilon^{1-\alpha} P_{\eta}(x+i\epsilon)$ are either {\em{both}} infinite, zero, or in $(0,+\infty)$.
\end{prop}
Proposition \ref{prop_deriv} will be used to analyze continuity with respect to Hausdorff measures of the measure $\kappa$ defined in (\ref{eq_defkappa}).

The following Theorem is key to spectral analysis of rank one perturbations. It provides a characterization of the components of $\eta$ in a Lebesgue decomposition. Proof can be found e.g. in \cite{A,I}.
\begin{theorem} \label{thm_charact}
Let $\eta$ be a Borel measure on the real line such that $\int \frac{1}{1+ y^2} \ud \eta(y) < \infty$. The following statements characterize the components in the Lebesgue decomposition of $\eta = \eta_{\mathrm{sing}} + \eta_{\mathrm{ac}}$:
\begin{itemize}
\item[(i)] $\ud \eta_{\mathrm{ac}}(x) = \frac{1}{\pi} P_{\eta}(x+i 0) \ud x$.
\item[(ii)] $\eta_{\mathrm{sing}}$ is supported on $\left\{x:P_{\eta}(x+i 0) = +\infty\right\}$.
\end{itemize}
\end{theorem}

Theorem \ref{thm_charact} implies a characterization of the spectral properties of the family $\{A_\lambda\}$. Out of this we shall only need the following statement related to the singular(pp+sc)-spectrum  of $\{A_\lambda\}$ (see \cite{A}, Theorem 12.2)
\begin{prop} \label{prop_speccharact}
\begin{itemize}
\item[(i)] $\mu_{\lambda,\mathrm{sing}}$ is supported on the set $\{x: F_\mu(x+i0) = -\frac{1}{\lambda}\}$.
\item[(ii)] The family of measures $\{\ud \mu_{\lambda,\mathrm{sing}}\}$ are mutually singular. In particular, a point $x \in \mathbb{R}$ can be an atom for at most one value of $\lambda$.
\end{itemize}
\end{prop}

\section{Spectral averaging} \label{sec_specaveg}

For a fixed $\sigma$-finite Borel measure $\nu$, consider the measure $\kappa$ introduced in  (\ref{eq_defkappa}). $\kappa$ is well defined since for any polynomial $p(x)$, $\left \langle \phi, p(A_\lambda) \phi \right \rangle$ is a polynomial in $\lambda$. Stone-Weierstra{\ss} and functional calculus then imply that $\lambda \mapsto \mu_\lambda (B)$ is Borel measurable for any Borel set $B \subseteq \mathbb{R}$.

We start our analysis of the continuity of $\kappa$ in relation to the continuity of $\nu$ with the following simple observation:
\begin{theorem} \label{thm_cont}
If $\nu$ is continuous, so is $\kappa$.
\end{theorem}
\begin{proof}
Apply part (ii) of Proposition $\ref{prop_speccharact}$ to $\kappa(\{x\})=\int \mu_\lambda(\{x\}) \ud \nu(\lambda)$, $x \in \mathbb{R}$.
\end{proof}

The following simple relation between the Poisson transforms of $\kappa$ and $\nu$ is crucial to further analyze the continuity properties of $\kappa$. 
\begin{prop} \label{prop_relPoisson}
Assume $\int \frac{1}{1+y^2} \ud \nu(y) < \infty$. Then,
\begin{equation} 
P_\kappa(z) =  P_\nu\left(-\frac{1}{F_\mu(z)}\right)  ~\mbox{,}
\end{equation}
for $z \in \mathbb{H}^+$.
\end{prop}
\begin{proof}
Using the definition of $\kappa$ in (\ref{eq_defkappa}), the monotone convergence Theorem implies
\begin{equation*}
\int f(x) \ud \kappa(x) = \int \left \{ \int f(x) \ud \mu_\lambda (x) \right \} \ud \nu(\lambda) ~\mbox{,}
\end{equation*}
for any measurable $0 \leq f$.

In particular for $z \in \mathbb{H}^+$,
\begin{eqnarray}
P_\kappa(z) & = & \int P_{\mu_\lambda}(z) \ud \nu(\lambda) \nonumber \\
= \int \dfrac{P_\mu(z)}{\vert 1+\lambda F_\mu(z)\vert ^2} \ud \nu(\lambda) & = & P_\nu\left(-\frac{1}{F_\mu(z)}\right) ~\mbox{.}
\end{eqnarray}
Here, the second equality follows from the Aronszajn-Krein formula \cite{A}
\begin{equation}
F_{\mu_\lambda}(z) = \dfrac{F_\mu(z)}{1+\lambda F_\mu(z)} ~\mbox{,}
\end{equation}
which relates the Borel transforms of the spectral measures $\mu_\lambda$ and $\mu$. 
\end{proof}
\begin{remark}
If $\nu$ is a finite measure an analogous result between the respective Borel transforms was first obtained in \cite{F}:
\begin{equation} \label{eq_relBorel}
F_\kappa(z) =  F_\nu\left(-\frac{1}{F_\mu(z)}\right) 
\end{equation}
Note that for non-finite $\nu$ the Borel transform will in general not exist (e.g. take $\nu$ to be the Lebesgue measure). In fact for $\sigma$-finite $\nu$, often the Poisson transform exists whereas its Borel transform does not. In these cases we still have a relation between the Poisson transforms of $\nu$ and $\kappa$ as established in Proposition \ref{prop_relPoisson}.
\end{remark}

In order to prove finer statements on the continuity of $\kappa$, we first establish some results for uniformly H\"older continuous $\nu$. Recall the following definition:
\begin{definition} \label{def_UHc}
Let $\eta$ be a $\sigma$-finite Borel measure on the real line and $\alpha \geq 0$. $\eta$ is uniformly $\alpha$ H\"older continuous (U$\alpha$H) if for some constant $K$, $\eta(I) \leq K \vert I \vert^\alpha$ for any interval $I$.
\end{definition}

\begin{remark}
\begin{itemize}
\item[(i)] U1H implies absolute continuity.
\item[(ii)] Using the Rogers-Taylor decomposition Theorem (see Theorem \ref{thm_rt1}), there are no {\em{non-trivial}} U$\alpha$H measures for $\alpha > 1$.
\end{itemize}
\end{remark}

For $\nu$ U$\alpha$H, Proposition \ref{prop_relPoisson} implies the following key estimate for the Poisson transform of $\kappa$:
\begin{prop} \label{prop_estim}
If $\nu$ is U$\alpha$H, $0 \leq \alpha \leq 1$, then for some constant $C_\alpha$ and all $z \in \mathbb{H}^+$
\begin{equation}
P_\kappa(z) \leq C_\alpha \left ( \dfrac{\vert F_\mu (z) \vert ^2}{P_\mu(z)} \right )^{1-\alpha} ~\mbox{.}
\end{equation}
In particular, $\int \frac{1}{1+x^2} \ud \kappa (x) < \infty$, whence $\kappa$ is a locally finite Borel measure (i.e. finite on compact sets).
\end{prop}

\begin{proof}
Let $z \in \mathbb{H}^+$. Recasting $P_\nu(z)$ in terms of the Lebesgue-Stilties measure induced by $M_\nu(\re{\{z\};\delta})$, we get
\begin{eqnarray} \label{eq_estim2}
P_\nu(z) & = & \im{\{z\}} \int_{0}^{+\infty} \dfrac{\ud M_\nu(\re{\{z\};\delta})}{\delta^2 + \im{\{z\}}^2} \nonumber \\
= \im{\{z\}} \int_{0}^{+\infty} 2 \delta \dfrac{M_\nu(\re{\{z\};\delta})}{\left(\delta^2 + \im{\{z\}}^2\right)^2} \ud \delta & \leq & \im{\{z\}} 2 K \int_{0}^{+\infty} \dfrac{\delta^{\alpha+1}}{\left(\delta^2+ \im{\{z\}}^2\right)^2} \ud \delta  \nonumber \\
& = & \dfrac{\pi \alpha K}{2 \sin\left(\frac{\pi \alpha}{2}\right)} \left(\im{\{z\}}\right)^{\alpha-1} ~\mbox{.}
\end{eqnarray}
Here, the second equality follows using integration by parts; the last equality is obtained by contour integration. For $\alpha=0$, the last equality in (\ref{eq_estim2}) is to be interpreted in the limit $\alpha \to 0$, i.e. $P_\nu(z) \leq K \im\{z\}^{-1}$. 

In particular, for $\nu$ U$\alpha$H, (\ref{eq_estim2}) establishes $\int \frac{1}{1+x^2} \ud \nu(x) < \infty$. Application of Proposition \ref{prop_relPoisson} hence yields the desired estimate.
\end{proof}

Proposition \ref{prop_estim} reveals the special nature of the case $\alpha=1$. Then, we obtain a {\em{uniform}} upper bound for the Poisson transform of $\kappa$ valid in all of $\mathbb{H}^+$. In fact, this implies for $\kappa$ to inherit ``full'' continuity of the measure $\nu$:

\begin{theorem} \label{thm_U1H}
If $\nu$ is U1H, so is $\kappa$.
\end{theorem}
\begin{proof}
Let $0 \leq f$ be continuous of compact support. Using Proposition \ref{prop_estim},
\begin{eqnarray}
\int f(x) \ud x & \geq & \limsup_{\epsilon \to 0^+} \frac{1}{C_1} \int f(x) P_\kappa(x+i\epsilon) \ud x \nonumber \\
& = & \limsup_{\epsilon \to 0^+} \frac{1}{C_1} \int \left( \int f(x) \dfrac{\epsilon}{(x-y)^2 + \epsilon^2} \ud x \right) \ud \kappa(y) \nonumber \\
& \geq & \frac{1}{C_1} \int \left(\lim_{\epsilon \to 0^+} \int f(x) \dfrac{\epsilon}{(x-y)^2 + \epsilon^2} \ud x \right) \ud \kappa(y) \nonumber \\
& = & \frac{\pi}{C_1} \int f(y) \ud \kappa(y) ~\mbox{.}
\end{eqnarray}
Here, the second equality follows from Tonelli, whereas the second inequality uses Fatou's Lemma. Note that $\sigma$-finiteness of $\kappa$ is implied by Proposition \ref{prop_estim}.
\end{proof}

Theorem \ref{thm_U1H} in particular implies $\ud \kappa(x) \ll \ud x$. Spectral averaging now arises as a special case where the density of $\kappa$ can be calculated explicitly.
\begin{proof}[Proof of Theorem \ref{thm_kotani}]
Since the Poisson transform of the Lebesgue measure \newline $P_\mathrm{Leb}(z) = \pi$, all $z \in \mathbb{H}^+$, Theorem \ref{thm_charact}(i) and Proposition \ref{prop_relPoisson} yield $\ud \kappa(x)=\ud x$.
\end{proof}

\section{Continuity with respect to Hausdorff measures} \label{sec_conthd}

In this section we analyze the degree of continuity of $\kappa$ induced by measures $\nu$ with lesser degree of continuity than considered in the previous section. To this end we make the following definitions:
\begin{definition} \label{def_alphacont}
For $0 \leq \alpha$ let $h^\alpha$ denote the $\alpha$-dimensional Hausdorff measure on $\mathbb{R}$. Let $\eta$ be a Borel measure on the real line.
\begin{enumerate}
\item $\eta$ is called $\alpha$-continuous ($\alpha$c) if $\eta(B)=0$ whenever $h^\alpha(B)=0$.
\item $\eta$ is called $\alpha$-singular if $\eta$ is supported on a set of zero measure $h^\alpha$ .
\end{enumerate}
\end{definition}

The main tools for proving Theorem \ref{thm_main} are the following two results due to Rogers \& Taylor \cite{J,K,L}:

\begin{theorem}[Rogers \& Taylor - 1 (see Theorem 67 in \cite{J})] \label{thm_rt1}
Let $\eta$ be a $\sigma$-finite Borel measure on $\mathbb{R}$ and $0 \leq \alpha \leq 1$. Consider the sets $T_{\eta;0+}^\alpha := \left\{x : 0 \leq \overline{D}_{\eta}^{\alpha}(x) < \infty \right\}$ and $T_{\eta;\infty}^\alpha := \left\{x : \overline{D}_{\eta}^{\alpha}(x) = \infty \right\}$. Then, $T_{\eta;0+}^\alpha$ and $T_{\eta;\infty}^\alpha$ are Borel measurable and
\begin{itemize}
\item[(i)] $\eta$ is $\alpha$c on $T_{\eta;0+}^\alpha$ ~\mbox{,}
\item[(ii)] $h^\alpha\left(T_{\eta;\infty}^\alpha\right)=0$ and $\eta$ is $\alpha$-singular on $T_{\eta;\infty}^\alpha$ ~\mbox{.}
\end{itemize}
\end{theorem}
\begin{theorem}[Rogers \& Taylor - 2 (see Theorem 68 in \cite{J})] \label{thm_rt2}
Let $\eta$ be $\sigma$-finite and $\alpha$-continuous, $\alpha \geq 0$. For $\epsilon > 0$, there exist mutually singular measures $\eta_1$ and $\eta_2$ with $\eta = \eta_1 + \eta_2$ such that
\begin{itemize}
\item[(i)] $\eta_1$ is U$\alpha$H and
\item[(ii)] $\eta_2(\mathbb{R}) < \epsilon$.
\end{itemize}
\end{theorem}
\begin{remark}
Depending on $\overline{D}_{\eta}^{\alpha}$, Theorem \ref{thm_rt1} decomposes $\eta$ into an $\alpha$c and an $\alpha$-singular component. It thus generalizes the usual Lebesgue decomposition for $\alpha=1$. The relevance of the Rogers Taylor decomposition in spectral theory was pointed out by Last, see \cite{M}.
\end{remark}

By Theorem \ref{thm_rt2}, any $\alpha$-continuous measure is almost U$\alpha$H. Hence, the proof of Theorem \ref{thm_main} boils down to establishing the statement for a U$\alpha$H measure $\nu$. To this end we shall use the following Lemma, which quantifies the asymptotic growth of $P_\eta$ and $Q_\eta$ near the support of a probability measure $\eta$.
\begin{lemma} \label{lem}
Let $\eta$ be a probability (Borel) measure on $\mathbb{R}$, then for $x \in \mathbb{R}$ and $\epsilon >0$
\begin{equation*} 
\max\left\{P_\eta(x+i\epsilon),\abs{Q_\eta(x+i\epsilon)}\right\} \leq \frac{2}{\epsilon} \sum_{n=0}^{\infty} 2^{-n} M_\eta(x;2^{n+1}\epsilon) ~\mbox{.}
\end{equation*}
\end{lemma}
\begin{proof}
Let $x \in \mathbb{R}$ and $\epsilon > 0$.
\begin{eqnarray*}
& \abs{Q_\eta(x+i\epsilon)} \leq \sum_{n=1}^{\infty} \int_{\epsilon 2^n \leq \abs{x-y} \leq \epsilon 2^{n+1}} \dfrac{\abs{x-y}}{(x-y)^2+\epsilon^2} \ud \eta(y) \\ 
& + \int_{\abs{x-y} \leq 2\epsilon} \dfrac{\abs{x-y}}{(x-y)^2+\epsilon^2} \ud \eta(y) \leq \frac{2}{\epsilon}\sum_{n=0}^{\infty} 2^{-n} M_\eta(x;2^{n+1}\epsilon) ~\mbox{.}
\end{eqnarray*}
By a similar computation we obtain the same upper bound for $P_\eta(x+i\epsilon)$.
\end{proof}
We note that this result is an extended version of a Lemma in \cite{F} for U$\alpha$H $\eta$.

Together with Proposition \ref{prop_estim}, Lemma \ref{lem} allows to control the asymptotic behaviour of $P_\kappa(x+i\epsilon)$ as $\epsilon \to 0^+$. We are thus in the position to prove Theorem \ref{thm_main}.

\section{Proof of the main Theorem (Theorem \ref{thm_main})} \label{sec_main}

We shall divide the proof into two steps; step 1 establishes the statement for $\nu$ U$\alpha$H. Theorem \ref{thm_rt2} then allows to extend the result to the $\alpha$c case (step 2).
\begin{description}
\item[Step 1] Assume $\nu$ to be U$\alpha$H. If $\alpha=1$, the statement follows directly from Theorem \ref{thm_U1H}. Let $\alpha < 1$. We first examine the situation outside the support of the measure $\mu$. 
\begin{prop} \label{prop_outsupp}
Let $0 < \alpha < 1$ and $\nu$ U$\alpha$H. Then $\kappa$ is $\alpha$c outside $\mathrm{supp} \mu$.
\end{prop}
\begin{proof}
Fixing $x \notin \mathrm{supp}\mu$, there exist positive constants $\Gamma_1$ and $\Gamma_2$ such that
\begin{equation*}
\abs{F_\mu(x+i\epsilon)}^2 \leq \Gamma_1  ~\mbox{,} ~ P_\mu(x+i\epsilon) \geq \Gamma_2 \epsilon ~\mbox{,}
\end{equation*}
for all $\epsilon >0$ sufficiently small. Hence by Proposition \ref{prop_estim} we obtain,
\begin{eqnarray} 
\epsilon^{1-\alpha} P_\kappa(x+i\epsilon) \leq C_\alpha \epsilon^{1-\alpha} \left ( \dfrac{\vert F_\mu (z) \vert ^2}{P_\mu(z)} \right )^{1-\alpha} 
\leq  C_\alpha \left( \dfrac{\Gamma_1}{\Gamma_2} \right)^{1-\alpha} \mbox{,}
\end{eqnarray}
which implies the claim by Theorem \ref{thm_rt1}.
\end{proof}
\begin{remark}
By Theorem \ref{thm_rt2} (see the argument given in step 2), the statement of Proposition \ref{prop_outsupp} remains valid if $\nu$ is (only) $\alpha$c.
\end{remark}

In order to analyze the situation within the support of $\mu$, we first establish the following Lemma:
\begin{lemma} \label{lem_insupp}
Let $0 < \alpha < 1$ and $\nu$ U$\alpha$H. Fix $0 < \beta < 1$. Then, $\kappa$ is $\gamma$c on the set $T_{\mu;0+}^{\beta}$ where
\begin{equation} \label{eqn_gamma}
\gamma(\alpha,\beta) = \alpha - 2 (1-\beta)(1-\alpha) ~\mbox{,}
\end{equation}
as long as $\beta > \max\left\{0,\frac{2-3\alpha}{2(1-\alpha)}\right\}$.
\end{lemma}

\begin{proof}
Let $\beta < 1$ be fixed. By Proposition \ref{prop_outsupp} the statement is true outside $\mathrm{supp}\mu$. Let $x \in \mathrm{supp}\mu$ and assume $\overline{D}_{\mu}^{\beta}(x) < \infty$ so that $M_\mu(x;\delta) \leq \Lambda_x \delta^\beta ~\mbox{,} ~\forall \delta > 0$. Thus,
\begin{eqnarray} \label{eq_mainthm_1}
\frac{2}{\epsilon} \sum_{n=0}^{\infty} 2^{-n} M_\mu(x;2^{n+1} \epsilon) \leq \Lambda_x 2^{1+\beta} \epsilon^{\beta-1} \sum_{n=0}^{\infty} 2^{-n(1-\beta)} < \infty ~\mbox{.}
\end{eqnarray}
Note that finiteness of the upper bound in (\ref{eq_mainthm_1}) requires $\beta < 1$. 

Let $\gamma < 1$. Using Proposition \ref{prop_estim} and Lemma \ref{lem}, estimate (\ref{eq_mainthm_1}) yields
\begin{eqnarray}
\epsilon^{1-\gamma} P_\kappa(x+i\epsilon) \leq B_{x,\beta} \left(\dfrac{\epsilon^{2(\beta-1) + \frac{1-\gamma}{1-\alpha}}}{P_\mu(x+i\epsilon)}\right)^{1-\alpha} ~\mbox{.}
\end{eqnarray}

By Theorem \ref{thm_rt1} and Proposition \ref{prop_deriv}, $\kappa$ will be $\gamma$c on the set \newline $\left \{ x : \limsup_{\epsilon \to 0^+} \epsilon^{1-\gamma} P_\kappa(x+i\epsilon) < \infty \right\}$. Choose $\gamma$ such that $2(\beta-1) + \frac{1-\gamma}{1-\alpha}=1$, i.e. $\gamma = \alpha - 2(1-\beta)(1-\alpha)$. Since, $\epsilon^{-1} P_\mu(x+i\epsilon) \to \int \frac{1}{(x-y)^2} \ud\mu(y)$ as $\epsilon \to 0^+$ and $\int \frac{1}{(x-y)^2} \ud\mu(y) > 0$ for $x \in \mathrm{supp}\mu$, 
we obtain that $\kappa$ is $\gamma$c on the set $T_{\mu;0+}^{\beta}$ with $\gamma$ determined by (\ref{eqn_gamma}). Finally, $\gamma > 0$ is ensured by requiring $\beta > \max\left\{0,\frac{2-3\alpha}{2(1-\alpha)}\right\}$.
\end{proof}

In summary we now obtain the claim for  $\nu$ U$\alpha$H: Let $\delta = \alpha(1-\epsilon)$, $0 < \epsilon < 1$. It suffices to prove the statement for $\epsilon$ sufficiently small. Let $\beta$ such that $\gamma(\alpha,\beta) = \delta$, i.e. $\beta = 1 - \frac{\alpha}{2(1-\alpha)}\epsilon$. Choosing $\epsilon$ sufficiently small we can ensure that $\beta > \frac{2-3\alpha}{2(1-\alpha)}$ which is required to apply Lemma \ref{lem_insupp}. 

For such choice of $\epsilon$ and $\beta$, Lemma \ref{lem_insupp} implies that for any Borel set $B$ with $h^\delta(B)=0$,
\begin{eqnarray} \label{eq_mainthm_2}
\kappa(B) = \int \mu_{\lambda,\mathrm{sing}}(B \cap T_{\mu;\infty}^\beta) \ud \nu(\lambda) \leq \int \mu_{\lambda,\mathrm{sing}}(T_{\mu;\infty}^1) \ud\nu(\lambda) = 0 ~\mbox{.}
\end{eqnarray}
Applying Proposition \ref{prop_speccharact} and \ref{prop_deriv}, $\mu_{\lambda,\mathrm{sing}}(T_{\mu;\infty}^1)=0$ for $\lambda \neq 0$, which by continuity of $\nu$ implies the last equality in (\ref{eq_mainthm_2}).

\item[Step 2] Let $0 < \alpha < 1$ and $\delta < \alpha$. If $\nu$ is $\alpha$c, then by Theorem \ref{thm_rt2} given $\epsilon>0$ there are measures $\nu_1 \perp \nu_2$, $\nu = \nu_1 + \nu_2$, such that $\nu_1$ is U$\alpha$H and $\nu_2(\mathbb{R}) < \epsilon$. Let $B \subseteq \mathbb{R}$ be a Borel set with  $h^\delta(B)=0$. Then, $\int \mu_\lambda(B) \ud \nu_1(\lambda) = 0$ by step 1, whence 
\begin{equation*}
\kappa(B) = \int \mu_\lambda(B) \ud \nu_2(\lambda) < \epsilon ~\mbox{.}
\end{equation*}
An analogous argument shows that $\kappa$ is absolutely continuous if $\alpha=1$, which concludes the proof of Theorem \ref{thm_main}.
\end{description}

\bibliographystyle{amsplain}

\end{document}